\documentclass[11pt]{amsart}
\usepackage{amsmath,stmaryrd,amsthm, tikz, calc}
\usepackage{color}
\usepackage[margin=1.25in]{geometry}
\date{\today}

\newcommand{\E}{\mathbb{E}}
\newcommand{\hE}{\hat{\mathbb{E}}}

\newcommand{\Cx}{\mathbb{C}}

\renewcommand{\H}{\mathcal{H}}

\newcommand{\cB}{\mathcal{B}}

\newcommand{\cE}{\mathcal{E}}

\newcommand{\cG}{\mathcal{G}}
\newcommand{\cH}{\mathcal{H}}
\newcommand{\cK}{\mathcal{K}}
\newcommand{\cL}{\mathcal{L}}
\newcommand{\cV}{\mathcal{V}}

\newcommand{\spec}{\mathop{\rm spec}}
\renewcommand{\ker}{\mathop{\rm ker}}
\newcommand{\ran}{\mathop{\rm ran}}

\newcommand{\Tr}{{\rm Tr}}

\newcommand{\be}{\begin{equation}}
\newcommand{\ee}{\end{equation}}
\newcommand{\bea}{\begin{eqnarray}}
\newcommand{\eea}{\end{eqnarray}}
\newcommand{\beann}{\begin{eqnarray*}}
\newcommand{\eeann}{\end{eqnarray*}}
\newcommand{\eq}[1]{(\ref{#1})}

\newcommand{\ua}{\uparrow}
\newcommand{\da}{\downarrow}

\newtheorem{thm}{Theorem}[section]
\newtheorem{proposition}[thm]{Proposition}
\newtheorem{Corollary}[thm]{Corollary}

\newtheorem{lemma}[thm]{Lemma}

 \numberwithin{equation}{section}

\renewcommand{\epsilon}{\varepsilon}

\newcommand{\onb}{orthonormal basis}
\newcommand{\ket}[1]{\left\vert{#1}\right\rangle}
\newcommand{\bra}[1]{\left\langle{#1}\right\vert}
\newcommand{\ketbra}[2]{|{#1}\rangle \! \langle {#2}|}
\newcommand{\braket}[2]{\langle {#1}, {#2}\rangle}
\newcommand{\expval}[3]{\langle {#1}\!\mid\!{#2}\! \mid\!{#3}\rangle}
\newcommand{\norm}[1]{\|{#1}\|}

 \def\idty{{\mathchoice {\mathrm{1\mskip-4mu l}} {\mathrm{1\mskip-4mu l}} %
{\mathrm{1\mskip-4.5mu l}} {\mathrm{1\mskip-5mu l}}}}

\begin{document}

\title{A class of two-dimensional AKLT models with a gap}

\author[H. Abdul-Rahman]{Houssam Abdul-Rahman}
\address{Department of Mathematics\\ 
University of Arizona\\
Tucson, AZ 85721, USA}
\email{houssam@math.arizona.edu}
\author[M. Lemm]{Marius Lemm}
\address{Department of Mathematics\\
Harvard University\\
Cambridge, MA 02138, USA}
\email{mlemm@math.harvard.edu}
\author[A. Lucia]{Angelo Lucia}
\address{Walter Burke Institute for Theoretical Physics and Institute for Quantum Information \& Matter \\
California Institute of Technology\\
Pasadena, CA 91125, USA}
\email{alucia@caltech.edu}
\author[B. Nachtergaele]{Bruno Nachtergaele}
\address{Department of Mathematics and Center for Quantum Mathematics and Physics\\
University of California, Davis\\
Davis, CA 95616, USA}
\email{bxn@math.ucdavis.edu}
\author[A. Young]{Amanda Young}
\address{Department of Mathematics\\
University of Arizona\\
Tucson, AZ 85721, USA}
\email{amyoung@math.arizona.edu}

\begin{abstract}
The AKLT spin chain is the prototypical example of a frustration-free quantum spin system with a spectral gap above its ground state. Affleck, Kennedy, Lieb, and Tasaki also conjectured that the two-dimensional version of their model on the hexagonal lattice exhibits a spectral gap. In this paper, we introduce a family of variants of the two-dimensional AKLT model depending on a positive integer $n$, which is defined by decorating the edges of the hexagonal lattice with one-dimensional AKLT spin chains of length $n$. We prove that these decorated models are gapped for all $n\geq 3$.
\end{abstract}

\maketitle

\section{Introduction}\label{sec:introduction}

A central question concerning a quantum spin system is whether it is gapped or gapless. (We say a system is gapped if its Hamiltonian exhibits a uniform spectral gap above the ground state. Otherwise, it is gapless.) The existence of a spectral gap is known to have wide-ranging consequences for the system's low energy physics. For instance, the ground states of gapped Hamiltonians display exponential clustering \cite{HK06,NS06} and, in one dimension, they are known to satisfy various notions of bounded complexity  \cite{AKLV13,ALVV18,H07,LVV15}. Of particular interest are the spin liquid states conjectured to describe a number of interesting 
two- and three-dimensional systems \cite{KRS87,W91,CL98}.
Moreover, with the advent of Hastings' spectral flow \cite{H04} (also called quasi-adiabatic evolution), it has become possible to explore
gapped ground state phases in considerable detail. Different gapped phases are separated from each other
by quantum phase transitions, which are accompanied by a closing of the spectral gap \cite{BMNS12,BHM10}. 
Accordingly, numerous recent works are concerned with the stability of the spectral gap under finite-range perturbations assuming local topological order \cite{BHM10,MZ13,MN18,NSY18}. From these considerations, it would be desirable to have a multitude of gapped Hamiltonians that one can use as starting points for further analysis. However, proving the existence of a spectral gap is a non-trivial mathematical task and there exist only limited tools \cite{FNW92,GM16,KL18,K88,LM18} and only a few special models in which a spectral gap has been rigorously established \cite{AKLT88,BHNY15,BN14,B17,BNY16,BG15,K88,LM18}, particularly in dimensions $\geq 2$. We also mention in passing that deciding whether a general Hamilltonian is gapped or not is known to be undecidable in general, even for reasonable (i.e., translation-invariant and local) Hamiltonians \cite{BCLP18,CPW15}. See also \cite{M17}.

The foundational work in the field was done by Affleck, Kennedy, Lieb and Tasaki (AKLT in the following) in 1988 \cite{AKLT87,AKLT88}. Motivated by a famous conjecture of Haldane \cite{H83a,H83b} that predicts a spectral gap for the one-dimensional integer-spin Heisenberg antiferromagnet, AKLT provided two main contributions that proved seminal in the years to come:

\begin{enumerate}
\item They defined what is now called the one-dimensional AKLT chain: a spin-$1$, isotropic (i.e., $SU(2)$-invariant) antiferromagnet on a one-dimensional chain. They found that it has a unique ground state in the thermodynamic limit and rigorously established a spectral gap.
 \item They defined analogous spin-$z/2$ AKLT models on any $z$-regular bipartite graph. They focused on the hexagonal lattice case (so spin-$3/2$) and derived exponential decay of correlations in the infinite-volume ground state. These facts led AKLT to conjecture that the hexagonal model is also gapped.
\end{enumerate}

We make two further remarks about the AKLT model: (a) The ground state of the hexagonal AKLT model was proved to have a unique \cite{KLT88} thermodynamic limit when the limit is taken with boundary conditions in a certain natural class. It is also known that the ground state in finite-volume with periodic boundary conditions (a finite honeycomb lattice wrapped on a torus) is unique
\cite{KK89}. Kennedy, Lieb, and Tasaki also proved exponential decay of spin-spin correlations, which is significant because it shows that the AKLT antiferromagnet does not exhibit N{\'e}el order, in contrast to its spin-$3/2$ Heisenberg analog, and that there very likely is a spectral gap above the ground state. (b) Historically, the AKLT chain provided the first example of a Hamiltonian whose ground states are matrix product states. This notion has been vastly generalized, starting with \cite{FNW92}, to what are now called tensor network states, and has developed into a central tenet of modern many-body physics \cite{O14,SCP10}.

The conjecture of AKLT that the hexagonal AKLT model (or any other AKLT model in dimension $\geq 2$) is gapped remains open to this day. This is insofar remarkable as all of the AKLT models share a key feature that makes the spectral gap problem in principle more amenable: They are frustration-free, meaning that the global ground state is also locally energy-minimizing.

In the present work, we introduce a novel family of AKLT models on `decorated' hexagonal lattices depending on an integer parameter $n$, and prove that these models are gapped for sufficiently large values of $n$. We call these models the {\em edge-decorated AKLT models} (or {\em decorated AKLT models} for short). The positive integer $n$ will be called the {\em decoration number}, and we explain its role in the next paragraph.

The decorated AKLT model is defined by replacing each edge of the hexagonal lattice with a copy of the one-dimensional AKLT chain of length $n$. Notice that this means that there are two types of vertices in the system: vertices of the hexagonal lattice which have degree $3$ and spin $3/2$, and ``internal vertices'' of the decorated edges which have degree $2$ and spin $1$; see Figure \ref{fig:DecoratedLattice}. The heuristic behind this construction is that the decorated AKLT model incorporates features of the one-dimensional AKLT chain, which is known to be gapped from the work of AKLT. While the decorated model is a bit contrived, it is not unreasonable to expect its ground state(s) to belong to the same gapped phase as those of the original AKLT model on the hexagonal lattice. It seems likely that the same features that generated interest in two-dimensional AKLT models \cite{PDC18,PSPC12,WAR11} are also present for the decorated AKLT model. In particular, we mention \cite{WHR14}, where it was shown that the valence-bond ground states of similarly decorated AKLT models can serve as a universal resource for quantum computation. Going beyond AKLT-type models, an $SU(3)$ spin liquid with $\mathbb{Z}_3$ topological order has recently been proposed in \cite{KVS18}. It too is expected to be gapped.

The strategy that we use to derive the spectral gap of the decorated model follows two main steps. We begin by taking the square of the Hamiltonian, as usual. Step 1 is to employ an inequality due to Fannes, Nachtergaele, and Werner \cite{FNW92} which relates the anticommutator between interaction terms to the angle between ground state projections via a duality argument. This inequality reduces the claim to a sufficiently strong bound on the angle between two ground state projections that overlap along one decorated edge (so mainly along a one-dimensional AKLT chain of length $n$). Step 2 is to establish the desired angle bound by a computation with quasi-one-dimensional matrix product states. With an eye toward possible future applications, we generalize the last computation (of the angle between ground state projections overlapping on a chain) to other models with matrix product ground states.

\section{AKLT models on decorated two-dimensional lattices}\label{sec:models}

For concreteness, we will first discuss in detail an AKLT model on a honeycomb lattice with additional spins on the edges. It will then be straightforward to consider generalizations to which the same arguments apply.

Let $\Gamma$ be the hexagonal lattice and $n\geq 1$. The standard AKLT model on $\Gamma$ \cite{KLT88} has a spin-$3/2$ degree of freedom at each vertex.
For the `decorated' models we introduce here, we add $n$ spin $1$'s along each edge of $\Gamma$ and call the resulting `lattice' $\Gamma^{(n)}$. On this graph with both degree 2 and
degree 3 vertices (see Figure \ref{fig:DecoratedLattice}) we define the AKLT Hamiltonian as usual with nearest neighbor interactions given by the orthogonal projection $P^{(z(e)/2)}$ onto the space of total
spin $z(e)/2$, where for any edge $e$, $z(e)$ is the sum  of the degrees of its two vertices. For the AKLT model on $\Gamma^{(1)}$ all the interaction terms are $P^{(5/2)}$.
For $n\geq 2$, the model also has interactions $P^{(2)}$ between neighboring spin $1$'s. This class of models is a special case of the general class of AKLT-type models studied in \cite{KK89}. There, it is shown that they are frustration-free and also that
the ground state is non-degenerate if the model is considered with periodic boundary conditions. The frustration-freeness is easily proved in the same way as for the original AKLT models by using the Valence Bond Solid construction of a non-zero vector in the kernel of the manifestly non-negative Hamiltonian.

\begin{figure}
\begin{tikzpicture}
\newcommand{\hexcoord}[2]
{[shift=(0:#1),shift=(60:#1),shift=(0:#2),shift=(-60:#2)]}
\foreach \x in {0,...,1}
\foreach \y in {0,1}{
\draw\hexcoord{\x}{\y}
(0:1)--(60:1)--(120:1)--(180:1)--(-120:1)--(-60:1)--cycle;
\foreach \n in {0, 60, 120, 180, -120, -60}{
\draw\hexcoord{\x}{\y} (\n:1)--(\n:1.3);
\draw[fill=black]\hexcoord{\x}{\y} (\n:1) circle (2.5pt);
\draw[fill=black]\hexcoord{\x}{\y} (\n+20:{sqrt((2*cos(\n)+cos(\n+60))^2/9+(2*sin(\n)+sin(\n+60))^2/9)}) circle (1.2pt);
\draw[fill=black]\hexcoord{\x}{\y} (\n+40:{sqrt((cos(\n)+2*cos(\n+60))^2/9+(sin(\n)+2*sin(\n+60))^2/9)}) circle (1.2pt);
}
}
\end{tikzpicture}
\caption{The decorated hexagonal lattice for $n=2$.}	
\label{fig:DecoratedLattice}
\end{figure}
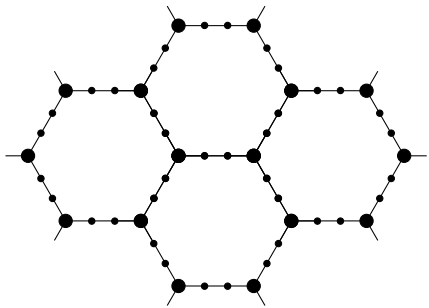

Let $\Lambda$ be a suitable finite subset of $\Gamma$ considered with periodic boundary conditions and denote by $\Lambda^{(n)}$ its decoration as above.
Let $\cE_{\Lambda^{(n)}}$ denote the set of edges of the decorated graph and consider
the Hamiltonian
\be
H_{\Lambda^{(n)}} = \sum_{e\in \cE_{\Lambda^{(n)}}} P^{(z(e)/2)}_e.
\ee
We claim that for $n$ large enough there is $\gamma_n>0$ such that for all nice $\Lambda$ the gap of $H_{\Lambda^{(n)}} $ is lower bounded by $\gamma_n$. In this context, `nice' $\Lambda$, means that one can consider the decorated graph $\Lambda^{(n)}$ as a union of overlapping 
subgraphs isomorphic to the H-shaped graph shown in Figure \ref{fig:AKLTG}. For concreteness, we will explicitly treat the case of periodic 
boundary conditions (a finite rectangle cut out of the hexagonal lattice and wrapped around a torus). Other shapes can be considered without any significant change in the arguments.

To prove the claim we consider a comparable model defined as follows. For each vertex $v$ in $\Lambda$ and the three edges meeting in $v$, we consider
the subsystem consisting of the spin $3/2$ at $v$ and the $3n$ spin $1$'s residing on the three edges.  Let $Y_v$ denote the corresponding set of $3n+1$
vertices in $\Lambda^{(n)}$ and define $h_v$ to be the AKLT Hamiltonian on $Y_v$. Then
\be
H_{\Lambda^{(n)}}\leq \sum_{v\in\Lambda} h_v \leq 2H_{\Lambda^{(n)}}.
\ee
To simplify things further, define $P_v$ to be the orthogonal projection onto $\ran h_v$.  It is a straightforward calculation
to check that $\ker P_v=\ker h_v$ is $8$-dimensional for $n=1$, and hence for all larger values of $n$ as well.
We will estimate the gap of
\be
\tilde H_{\Lambda^{(n)}} = \sum_{v\in\Lambda} P_v,
\ee
which is also comparable to $H_{\Lambda^{(n)}} $:
\be\label{comparable}
\frac{1}{2}\gamma_Y\tilde H_{\Lambda^{(n)}} \leq H_{\Lambda^{(n)}} \leq \Vert h_Y\Vert \tilde H_{\Lambda^{(n)}} ,
\ee
with $\gamma_Y>0$. This inequality implies that $\tilde{H}_{\Lambda^{(n)}}$ is also frustration-free since $\ker(\tilde{H}_{\Lambda^{(n)}}) =\ker(H_{\Lambda^{(n)}})\neq \{0\}$. Therefore, it suffices to study the gap of $\tilde H_{\Lambda^{(n)}} $. 

We will obtain a lower bound for the gap of $ \tilde H_{\Lambda^{(n)}}$, by finding a constant  $\gamma>0$ satisfying
\be\label{alaknabe}
(\tilde H_{\Lambda^{(n)}})^2 = \tilde H_{\Lambda^{(n)}}  + \sum_{\{v,w\}\subset \Lambda, v\neq w} (P_vP_w+ P_w P_v )\geq  \gamma \tilde H_{\Lambda^{(n)}} .
\ee
If $v$ and $w$ are {\em not} nearest neighbors, $P_v$ and $P_w$ commute and $P_vP_w+ P_w P_v \geq 0$.
Therefore, in the second term we can drop all contributions from such pairs.
For the nearest neighbor pairs $(v,w)$, instead of the combinatorial style argument in \cite{K88},
which requires good estimates of a specific finite-volume gap, we apply the following inequality for a pair orthogonal projections $E$ and $F$ (for a proof see \cite[Lemma 6.3]{FNW92}):
\be\label{lowerbound}
EF + FE \geq -\Vert EF - E\wedge F\Vert (E+F).
\ee
Here,  $E\wedge F$ is the orthogonal projection onto $\ran E \cap \ran F$. We need this for $E=P_{v}, F=P_{w}$, with $v,w$ nearest neighbors in $\Lambda$.

The norm in \eq{lowerbound}  remains unchanged if we replace $E$ and $F$ by $\idty - E$ and $\idty -F$, which, in our application, are the corresponding ground state projections.
All nearest neighbor pairs are equivalent in this consideration and we denote $\idty- P_v=P_{\Yright}, \idty- P_w=P_{\Yleft}$, and $P_\Yright\wedge P_\Yleft=P_{\Yright\Yleft}$.
Define
\be\label{def_epsilon_n}
\epsilon_n=\Vert P_\Yright P_\Yleft - P_{\Yright\Yleft}\Vert.
\ee
Since every $v\in\Lambda$ has 3 nearest neighbors, we have shown
\be
(\tilde H_{\Lambda^{(n)}})^2 \geq (1-3\epsilon_n) \tilde H_{\Lambda^{(n)}}.
\ee
Therefore,
\be \label{gap_cond}
\mbox{gap} (H_{\Lambda^{(n)}}) \geq \frac{1}{2} \gamma_Y (1-3\epsilon_n).
\ee
It remains to show that $ \epsilon_n < 1/3$.

\begin{proposition}\label{prop:epsilonn}
Let
$$
A_n = \frac{4}{3^{n} \left( 1- \frac{8(1+3^{-2n-1})}{3^n(1-3^{-2n})}\right)}.
$$
Then, for all $n\geq 3$, the quantity $\epsilon_n$ defined in \eq{def_epsilon_n} satisfies
\be
\epsilon_n \leq A_n + A_n^2 \left(1+ \frac{8(1+3^{-2n-1})^2}{3^n(1-3^{-2n})^2}\right) < 1/3.
\ee
\end{proposition}

The proof of this proposition is contained in the next two sections. As a consequence, we can state the following theorem.

\begin{thm}\label{thm:DecSpecGap}
The spectral gap above the ground state of the AKLT model on the edge-decorated honeycomb lattice with $n\geq 3$ has a strictly positive lower bound uniformly for all finite volumes with periodic boundary conditions.
\end{thm}

\section{Ground state projections for subgraphs overlapping on a chain}\label{sec:bound}

To prove Proposition \ref{prop:epsilonn} we will formulate the quantity $\epsilon_n$ in terms of the ground state projections of quasi-one-dimensional Matrix Product States (MPS).
A rather straightforward generalization of  the arguments in  \cite{FNW92} to MPS systems with matrices that may vary from site to site and are not necessarily square, will
then yield the desired estimate. With an eye toward possible further generalizations and applications, we will estimate $\epsilon_n$ in a slightly more general setting, which we now introduce.
The role of the finite subgraph $\Yright\Yleft$ of the decorated honeycomb lattice will be played a finite graph $G$ with a certain structure and we will make a number of assumptions on
the ground states of a frustration free Hamiltonian on $G$. In Section \ref{sec:AKLT_Est} we will show that these assumptions are satisfied for the AKLT model on the decorated honeycomb lattice.

\subsection{Assumptions on the tensor network states for the local patch $G$.}\label{sec:assumptions}

Consider a finite graph $G=(\cV,\cE)$ of the form $G_L-C_n-G_R$, meaning there are finite graphs
$G_L$ and $G_R$, $C_n=[v_1,v_n]$ is a chain of $n$ vertices, and there exist $v_L\in G_L, v_R\in G_R$ such that $\cV$ is the disjoint union of the vertices of $G_L, C_n$ and $G_R$ and $\cE$
consists of the edges of $G_L, C_n$ and $G_R$ together with $(v_L,v_1)$ and $(v_n,v_R)$, see Figure~\ref{fig:AKLTG}.
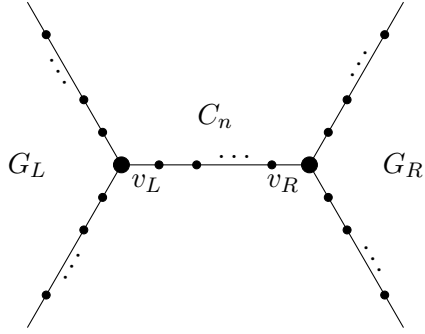
\begin{figure}
	\begin{tikzpicture}
	\node[yshift = .65cm] at (0:1.25){$C_n$};
	\node at (0:-1.25){$G_L$};
	\node at (0:3.75){$G_R$};
	\draw[fill=black] (0:0) circle (3pt) node [below right] {$v_L$};
	\draw[fill=black] (0:2.5) circle (3pt) node [below left] {$v_R$};
	\foreach \n in {0,120,240}{
		\draw(\n:0)--(\n:2.5);
		\node[rotate = \n, yshift=3pt] at (\n:1.5) {$\ldots$};
		\foreach \m in {.5, 1, 2}
		\draw[fill=black] (\n:\m) circle (1.5pt);
	}
	\foreach \n in {60,-60}
	{
	\draw[xshift=2.5cm] (\n:0)--(\n:2.5);	
	\node[xshift=2.5cm, rotate = \n, yshift=3pt] at (\n:1.5) {$\ldots$};
	\foreach \m in {.5, 1, 2}
	\draw[fill=black, xshift=2.5cm] (\n:\m) circle (1.5pt);
}
	\end{tikzpicture}
	\caption{The graph $G$ for the decorated AKLT model.}
	\label{fig:AKLTG}
\end{figure}
We consider a frustration-free Hamiltonian on $G$ of the following form:
\be\label{HG}
H_G = H_{G_L} + H_{C_n} + H_{G_R} + h_{v_L,v_1} + h_{v_n,v_R},
\ee
where $H_{G_L}$, $H_{G_R}$, and $H_{C_n}$ satisfy the following conditions. First, we assume $H_{C_n}$ has ground states given by a translation invariant MPS with a primitive transfer matrix $\E$.
Let $D$ denote the bond dimension of this MPS and pick an \onb\ $\{\ket{i}\mid 1\leq i \leq d\}$ for the physical degree of freedom at each of the $n$ sites of $C_n$. In this case the
transfer matrix $\E$, in isometric form, is given in terms of $d$ $D\times D$ matrices $V_i$:
\be
\E(B) = \sum_{i = 1}^d V_i^* B V_i, B\in M_D, \mbox{ with } \E(\idty) = \idty.
\ee
The primitivity assumption implies that there exists a non-singular density matrix $\rho \in M_D$ satisfying $\E^{t}(\rho) = \rho$, and constants $C\geq 0$ and $\lambda \in [0,1)$, such that
\be\label{an}
a(n) : = \left\| \E^n - \ketbra{\idty}{\rho} \right\| \leq C \lambda^n.
\ee
For the spin-1 AKLT chain one has this estimate with $C=1$ and $\lambda =1/3$, and both constants are sharp in that case.

We now turn to the assumptions we make on $H_{G_L}$ and $H_{G_R}$. For $\sharp\in\{L,R\}$, let $\cH_\sharp$ be the Hilbert space associated with the system on $G_\sharp$.
We assume that the ground states of $H_{G_\sharp}$ are given by a tensor $T^\sharp$ as follows. We will consider $T^L$ as a set
of $\dim\cH_L$ $D\times D_L$ matrices labeled by an \onb\ $\{ \ket{r}_L\}$ of $\cH_L$  and $T^R$ as a set of $\dim\cH_R$ $D_R\times D$ matrices labeled by an \onb\ $\{ \ket{r}_R\}$ of $\cH_R$.
The physical Hilbert space for the system on $G$ is $\cH_G = \cH_{G_L} \otimes \cH_{C_n}\otimes \cH_{G_R}$, and the auxiliary space, which parametrizes the ground states,
is $\cK_G=\Cx^{D_L}\otimes\Cx^{D_R}$, which we identify
with $\cL(\Cx^{D_R},\Cx^{D_L})$ and equip with the standard inner product $\langle\cdot,\cdot\rangle_{\cK_G}$. 
The map $\Gamma_G: \cK_G \to \cH_G$ is then given by
\be\label{GammaG}
\Gamma_G(B) =  \sum_{l,i_1,\ldots,i_n,r}\Tr [B T^R_r  V_{i_n}\cdots V_{i_1} T^L_l]  \ket{l}_L\otimes \ket{i_1,\ldots,i_n}\otimes \ket{r}_R, B\in \cK_G.
\ee
We assume that $H_G$ is frustration free, which means that all terms in \eq{HG} are non-negative and $\ker H_G \neq \{0\}$ and, in addition, we assume
$\ker H_G = \ran \Gamma_G$.

We also introduce the transfer matrices associated with $G_L$ and $G_R$,
$\E_L:M_D\to M_{D_L}$ and $\E_R:M_{D_R}\to M_{D}$, as follows:
\be
\E_L(B) =  \sum_l (T^L_l)^* B  T^L_l,\quad \E_R(B) = \sum_r   (T^R_r)^* B T^R_r,
\label{transfer-matrix-L-R}
\ee
and define
\be
Q_L = \E_L(\idty),\quad  Q_R = \E_R^t(\rho)
\ee
where $\E_R^t(B) =  \sum_r   T^R_r B (T^R_r)^*.$ We assume that $Q_L$ and $Q_R$ are non-singular.

For $\Lambda \in \{G_L, C_n, G_R, G_L-C_n, C_n-G_R\}$, considered as subsets of $G$, let $\cH_\Lambda$ and $\cK_\Lambda$ denote the corresponding physical and auxiliary Hilbert spaces, respectively, and define the corresponding maps $\Gamma_\Lambda : \cK_\Lambda \to\cH_\Lambda$ in the obvious way. These maps are of the same form as $\Gamma_G$ in \eq{GammaG};
if one or both parts described by $G_L$ or $G_R$ are absent, the absent degrees of freedom associated with $G_\sharp$ correspond to taking $\cH_{G_\sharp} =\Cx$, $H_{G_\sharp} =0$, and $T^\sharp \equiv 1$.

We assume that the maps $\Gamma_\Lambda$ are injective, meaning
\be\label{Gamma_injective}
\dim \ran\Gamma_\Lambda  =\dim \cK_\Lambda,
\ee
and $ \ran\Gamma_\Lambda = \ker H_\Lambda$.

\subsection{General estimate of $\epsilon_n$.}

Let $\cG_\Lambda$ denote $\ker H_\Lambda$, and $P_\Lambda$  the orthogonal projection onto $\cG_\Lambda$.
Our next goal is to estimate $\epsilon_n = \Vert P_{G_L-C_n}P_{C_n-G_R}-P_G\Vert$. It is easy to see that
$\epsilon_n$ is explicitly given by the following expression:
\be\label{innerproductformula}
\epsilon_n = \sup \left\{\left. \frac{\vert \braket{\phi}{\psi}\vert}{\Vert \phi\Vert \Vert \psi\Vert}\ \right\vert
\phi\in \cG_{G_L-C_n}\otimes \cH_{G_R},\ \psi \in \cH_{G_L}\otimes \cG_{C_n-G_R},\ \phi,\psi \perp \cG_G,\  \phi,\psi\neq 0\right\}.
\ee

We will derive an estimate of the type of inner products that appear in \eq{innerproductformula}, but first recall some basic properties of MPS.

Since $\rho$ is non-singular and positive, it defines an inner product on $M_D$ by
\begin{equation}
\langle A, B \rangle_{\rho} = \Tr \rho A^*B \quad \mbox{for all } A, B \in M_D \, ,
\end{equation}
and let $\Vert\cdot \Vert_\rho$ denote the corresponding norm. We will also let $\rho_{\rm min}$ denote the smallest
eigenvalue of $\rho$, which is positive by assumption. It follows that the norm $\Vert\cdot \Vert_\rho$ is equivalent to the
Hilbert-Schmidt norm on $M_D$, which is given by $\Vert \cdot\Vert_2=\sqrt{\Tr A^*A}$. Explicitly:
\begin{equation} \label{norms}
\| A \|_2 \leq  \frac{1}{\sqrt{\rho_{\rm min}}} \| A \|_{\rho},\quad A \in M_D \, .
\end{equation}

The map $\Gamma_{C_n}: M_D\to \H_{C_n}$ is explicitly given by
\be
\Gamma_{C_n}(B) =\sum_{i_1,\ldots,i_n}  \Tr [B V_{i_n}\cdots V_{i_1}]\ket{i_1,\ldots,i_n}, \ B\in M_D.
\label{GammaCn}\ee

\begin{lemma}[\expandafter{\cite[Lemma 5.2]{FNW92}}] \label{lem:innerproductMPS}
For any $B,C \in M_D$,
\begin{equation} \label{innerprod}
\left| \langle \Gamma_{C_n}(B), \Gamma_{C_n}(C) \rangle - \langle B, C \rangle_{\rho} \right|  \leq  a(n)  \Tr \rho^{-1}
\| B \|_{\rho} \| C \|_{\rho} \,.
\end{equation}
\end{lemma}
\begin{proof}
Using \eq{GammaCn} we can express the inner product as follows:
\bea \label{MPS_IP}
\langle \Gamma_{C_n}(B), \Gamma_{C_n}(C) \rangle  &=&
\sum_{i_1,\ldots,i_n} \overline{ \Tr [B V_{i_n}\cdots V_{i_1}]} \Tr[ C V_{i_n}\cdots V_{i_1}]\nonumber\\
&=&\sum_{i_1,\ldots,i_n} \Tr [V_{i_1}^*\cdots V_{i_n}^* B^*] \Tr[ C V_{i_n}\cdots V_{i_1}].
\eea
By expanding the traces using any \onb\  $\{\ket{1},\ldots,\ket{D}\}$ for $\Cx^D$, we obtain
\begin{eqnarray}
\langle \Gamma_{C_n}(B), \Gamma_{C_n}(C) \rangle & = & \sum_{\alpha,\beta=1}^D \sum_{i_1,\ldots,i_n} \bra{\alpha} V_{i_1}^*\cdots V_{i_n}^*B^*\ketbra{\alpha}{\beta}
 C V_{i_n}\cdots V_{i_1}\ket{\beta} \nonumber \\
& = & \sum_{\alpha,\beta=1}^D \bra{\alpha} \hE^n \left( B^* \ketbra{\alpha}{\beta}C \right) \ket{\beta}.
\label{Trace_to_TransferOp}
\end{eqnarray}
Now observe that
\be
\langle B,C\rangle_\rho = \sum_{\alpha,\beta=1}^D \bra{\alpha} \ketbra{\idty}{\rho} \left( B^* \ketbra{\alpha}{\beta}C \right) \ket{\beta}.
\ee
Combining these two expressions and using \eq{an}, we obtain
\beann
\left| \langle \Gamma_{C_n}(B), \Gamma_{C_n}(C) \rangle - \langle B, C \rangle_{\rho} \right| &\leq&
\sum_{\alpha,\beta=1}^D \vert \bra{\alpha} (\E^n-\ketbra{\idty}{\rho}) \left( B^* \ketbra{\alpha}{\beta}C \right) \ket{\beta}\vert \\
&\leq&  a(n) \left(\sum_{\alpha=1}^D \Vert B^*\ket{\alpha}\Vert\right)\left(  \sum_{\beta=1}^D\Vert C^*\ket{\beta}\Vert\right)
\eeann
Now, pick for the \onb\ one that diagonalizes $\rho$, such that $\rho \ket{\alpha} = \rho_\alpha \ket{\alpha}$. Then
$$
\left( \sum_{\alpha=1}^D \Vert B^*\ket{\alpha}\Vert \right)^2 = \left( \sum_{\alpha=1}^D \Vert B^*\ket{\alpha}\Vert \rho_\alpha^{1/2} \rho_\alpha^{-1/2}\right)^2
\leq  \sum_{\alpha=1}^D \rho_\alpha \bra{\alpha}BB^*\ket{\alpha}  \sum_{\alpha=1}^D \rho_\alpha^{-1} = \Vert B\Vert^2_\rho \Tr \rho^{-1}.
$$
Together with the analogous estimate for the second factor, this proves the lemma.
\end{proof}
Note that one has the bound
$$
\Tr \rho^{-1}\leq \frac{D}{\rho_{\rm min}},
$$
which is often saturated in models with symmetry. It will be convenient to define
\be\label{bn}
b(n) = a(n)\Tr \rho^{-1}
\ee
The following is an immediate corollary of Lemma \ref{lem:innerproductMPS}, and shows that $\Gamma_{C_n}$ is injective for sufficiently large $n$.

\begin{Corollary} \label{Cor:bd1} For any $B \in M_D$, the bound
\begin{equation}
\| B \|_{\rho} \sqrt{1-b(n)} \leq \| \Gamma_n(B) \| \leq \| B \|_{\rho} \sqrt{1+b(n)}
\end{equation}
holds for $n$ sufficiently large so that $b(n)\leq1$.
\end{Corollary}
\begin{proof}
The bound
\begin{equation}
\left| \| \Gamma_{C_n}(B) \|^2 - \| B \|_{\rho}^2 \right| \leq b(n) \| B \|_{\rho}^2
\end{equation}
follows immediately from (\ref{innerprod}). If $B=0$, there is nothing to prove. Otherwise, this bound
can be re-written as
\begin{equation}
-b(n) \leq \frac{\| \Gamma_{C_n}(B) \|^2}{ \| B \|^2_{\rho}} - 1 \leq b(n)
\end{equation}
{f}rom which the above claim readily follows.
\end{proof}

Inner products of vectors of the form $\Gamma_G(B), B\in \cK_G$ can be estimated with a straightforward generalization of Lemma \ref{lem:innerproductMPS}. To formulate the result,
for each $\Lambda\in\{G, \, G_L-C_n, \, C_n-G_R\}$ we define an inner product on $\cK_\Lambda$, denoted by $\langle\cdot,\cdot\rangle_\Lambda$, via
\bea
\langle B,C\rangle_G & = & \Tr (Q_R B^* Q_L C) \\
\langle B, C \rangle_{G_L-C_n} & = & \Tr (\rho B^* Q_L C) \\
\langle B, C \rangle_{C_n-G_R} & = & \Tr (Q_R B^* C)
\eea
That these are inner products follows from the positive-definiteness of $Q_L$ and $Q_R$. With respect to these inner products we obtain the following analog of Lemma~\ref{lem:innerproductMPS}.
\begin{lemma} \label{lem:IPEstimates}
Let $\Lambda\in\{G, \, G_L-C_n, \, C_n-G_R\}$. Then for any $B, C \in \cK_\Lambda$,
\be \label{innerLambda}
\left| \langle \Gamma_\Lambda(B), \Gamma_\Lambda(C) \rangle - \langle B, C \rangle_\Lambda\right|  \leq  a(n)D^2C_\Lambda \|B\|\|C\|,
\ee
where
\be
C_G = \|\E_L\|\|\E_R\|, \quad C_{G_L-C_n} = \|\E_L\|, \quad \text{and}\quad C_{C_n-G_R}=\|\E_R\|.
\ee
\end{lemma}

\begin{proof}
We prove the bound in the case of $\Lambda = G$. All other cases follow from similar arguments. Let $B, \, C \in \cK_G$, and $\{\ket{1}, \, \ldots, \, \ket{D}\}$ be an orthonormal bases of $\Cx^D$. Then, calculating similar to \eqref{MPS_IP} and \eqref{Trace_to_TransferOp}, we find
\bea
\langle \Gamma_{G}(B), \Gamma_{G}(C) \rangle
& = &
\sum_{\substack{\ell,r \\ i_1,\ldots,i_n}} \Tr [V_{i_1}^*\cdots V_{i_n}^*(T_r^R)^* B^*(T_\ell^L)^*] \Tr[ T_\ell^L C T_r^R V_{i_n}\cdots V_{i_1}] \nonumber\\
& = & \sum_{\alpha, \beta = 1}^D \bra{\alpha} \E^{n}\circ \E_R\big[B^*\E_L(\ketbra{\alpha}{\beta})C\big] \ket{\beta}, \label{IP_gen1}
\eea
where we have also use cyclicity of the trace in the first equality. Now consider $\braket{B}{C}_G$. It can easily be shown, e.g. by simplifying the RHS, that
\be \label{IP_gen2}
\langle B,C\rangle_G = \sum_{\alpha, \beta=1}^D \bra{\alpha} \ketbra{\idty}{\rho}\circ \E_R[B^*\E_L(\ketbra{\alpha}{\beta})C\big] \ket{\beta}.
\ee
By substituting these \eqref{IP_gen1} and \eqref{IP_gen2} into \eqref{innerLambda} and then using \eqref{an}, we estimate as follows:
\bea
\left| \langle \Gamma_G(B), \Gamma_G(C) \rangle - \langle B, C \rangle_G\right|
& \leq &
\sum_{\alpha, \beta = 1}^D
\left|\bra{\alpha} (\E^{n}-\ketbra{\idty}{\rho})\circ\E_R\big[B^*\E_L(\ketbra{\alpha}{\beta})C\big] \ket{\beta}\right| \nonumber \\
& \leq &
a(n) \sum_{\alpha, \beta = 1}^D \left\| \E_R\big[B^*\E_L(\ketbra{\alpha}{\beta})C\big] \right\| \nonumber\\
& \leq & a(n)D^2\|\E_L\|\|\E_R\| \|B\| \|C\|. \label{G_final_est}
\eea
This completes the claim.
\end{proof}

Note that the bound in Lemma~\ref{lem:IPEstimates} is expressed in terms of the operator norms $\| B \|$ and  $\| C \|$. This is just a common norm of reference. The natural norm to use is the one
induced by the inner product that appears on the left of \eqref{innerLambda}, as in done in Lemma \ref{lem:innerproductMPS}. Since these norms are all equivalent to the operator norm, the estimates from \eqref{innerLambda} can be converted to the `natural' norm by multiplying by an appropriate constant as follows: Let $q_L$ (resp. $q_r$) be the minimal eigenvalue of $Q_L$ (resp. $Q_R$). Then,
\be \label{norm_equivs}
\|B\|  \leq  \frac{1}{\sqrt{q_Lq_R}}\|B\|_G, \quad
\|B\|  \leq  \frac{1}{\sqrt{\rho_{\rm min}q_L}}\|B\|_{G_L-C_n}, \quad
\|B\|  \leq  \frac{1}{\sqrt{q_R}}\|B\|_{C_n-G_R}.
\ee
We can thus obtain a corollary to Lemma~\ref{lem:IPEstimates} similar to Corollary~\ref{Cor:bd1}.
\begin{Corollary} \label{Cor:bd2}
Let $\Lambda\in\{G, \, G_L-C_n, \, C_n-G_R\}$. Then for any $B \in \cK_\Lambda$,
\begin{equation}\label{eq:bd2}
\| B \|_{\Lambda} \sqrt{1-b_\Lambda(n)} \leq \|\Gamma_\Lambda(B) \| \leq \| B \|_{\Lambda} \sqrt{1+b_\Lambda(n)}
\end{equation}
holds for $n$ sufficiently large so that $b_\Lambda(n)\leq1$, where
\be
\label{left_right_bn}
    b_{G}(n) = \frac{a(n)D^2}{q_L q_R} \|\E_L\|\|\E_R\|, \quad
    b_{G_L-C_n}(n)  = \frac{a(n)D^2}{\rho_{\min} q_L}\|\E_L\|, \quad
    b_{C_n-G_R}(n) = \frac{a(n)D^2}{q_R}\|\E_R\|.
\ee
\end{Corollary}
In order to simplify notation, we will write $b_L(n)$ for $b_{G_L-C_n}(n)$ and $b_R(n)$ for $b_{C_n-G_R}(n)$.

Ultimately, we will want a bound for the inner product in \eq{innerproductformula} in terms of the norms of the vectors $\phi$ and $\psi$ defined below. Lemma~\ref{lem:IPEstimates} can also be used to show that this is once again straightforward at the cost of another prefactor in the bound.

Since $\Gamma_{G_L-C_n}$ and $\Gamma_{C_n-G_R}$ are assumed to be injective, there exist $D_R$ $D_L\times D$ matrices $B_{\phi}(r)$, and $D_L$ $D\times D_R$ matrices $B_{\psi}(l)$,
uniquely determined by $\phi$ and $\psi$, such that
\bea
\phi &=& \sum_{l,i_1,\ldots,i_n,r}  \Tr [B_\phi(r) V_{i_n}\cdots V_{i_1}T^L_l]  \ket{l}_L\otimes \ket{i_1,\ldots,i_n}\otimes \ket{r}_R \label{phi}\\
\psi &=& \sum_{l,i_1,\ldots,i_n,r}  \Tr [B_\psi(l) T^R_rV_{i_n}\cdots V_{i_1}]  \ket{l}_L\otimes \ket{i_1,\ldots,i_n}\otimes \ket{r}_R \label{psi}.
\eea
These expressions are simply expansions of
\bea
\phi &=& \sum_{r} \Gamma_{G_L-C_n}(B_{\phi}(r)) \otimes \ket{r}_R\label{phi1}\\
\psi &=& \sum_{l} \ket{l}_L\otimes  \Gamma_{C_n-G_R}(B_{\psi}(l)).\label{psi1}
\eea
It will be convenient to define $C_\phi,C_\psi\in \cK_G$ as follows:
\bea
C_{\phi} &=& \sum_{r} B_{\phi}(r) \rho (T^R_{r})^*\label{Cphi}\\
D_{\psi} &=& \sum_{l} (T^L_{l})^* B_{\psi}(l).\label{Dpsi}
\eea

Next, we consider inner products of the form $\langle \phi, \psi \rangle$, with $\phi \in \cG_{G_L-C_n} \otimes \cH_{G_R}$ and $\psi \in \cH_{G_L} \otimes \cG_{C_n-G_R}$.

\begin{lemma}\label{lem:innerproducts}
Suppose that $n$ is large enough so that $b(n) <1$. Then, for all $\phi \in \cG_{G_L-C_n} \otimes \cH_{G_R}$ and $\psi \in \cH_{G_L} \otimes \cG_{C_n-G_R}$,
we have
\begin{equation}
\left| \langle \phi, \psi \rangle - \langle C_{\phi}, D_{\psi} \rangle_{\cK_G} \right| \leq \frac{b(n)}{\sqrt{1-b_{LR}(n)}} \| \phi \| \| \psi \|,
\label{psi_phi_IP}
\end{equation}
where
\be
\label{b_LR}
b_{LR}(n) = b_{L}(n) + b_{R}(n) - b_{L}(n)b_{R}(n),
\ee
and $b_{L}(n)$ and $b_{R}(n)$ are defined in \eqref{left_right_bn}.
\end{lemma}
\begin{proof}
Using the expansions \eq{phi} and \eq{psi} we find
\begin{eqnarray}
\langle \phi, \psi \rangle & = & \sum_{l,i_1,\ldots,i_n,r}   \overline{\Tr [B_\phi(r) V_{i_n}\cdots V_{i_1}T^L_l]}\Tr [B_\psi(l) T^R_rV_{i_n}\cdots V_{i_1}] \nonumber \\
& = & \sum_{l,r} \langle \Gamma_{C_n} (T^L_l B_{\phi}(r)), \Gamma_{C_n} ( B_{\psi}(l) T_r^R) \rangle
\end{eqnarray}
Similarly, we observe that $\langle C_{\phi}, D_{\psi} \rangle_{\cK_G}$ can be expressed as a sum of inner products:
\beann
\langle C_{\phi}, D_{\psi} \rangle_{\cK_G} &=&
\Tr \left(\sum_r T^R_r   \rho B_{\phi}(r)^* \right) \left(\sum_l (T^L_l)^* B_{\psi}(l) \right)\\
&=& \sum_{l,r} \Tr \rho B_{\phi}(r)^* (T^L_l)^* B_{\psi}(l) T^R_r \\
&=&  \sum_{l,r} \langle T^L_l B_{\phi}(r), B_{\psi}(l) T^R_r \rangle_{\rho}.
 \eeann

Now we apply Lemma~\ref{lem:innerproductMPS} term by term to obtain, using Cauchy-Schwarz:
\begin{eqnarray}
\left| \langle \varphi, \psi \rangle -  \langle C_{\phi}, D_{\psi} \rangle_{\cK_G} \right| &
\leq & \sum_{l,r} \left| \langle \Gamma_{C_n} (T^L_l B_{\phi}(r)), \Gamma_{C_n} ( B_{\psi}(l) T_r^R) \rangle  - \langle T^L_l B_{\phi}(r), B_{\psi}(l) T^R_r \rangle_{\rho}  \right|
\nonumber \\
 & \leq & b(n) \sum_{l,r} \| T^L_l B_{\phi}(r) \|_{\rho} \cdot \| B_{\psi}(l) T^R_r\|_{\rho} \nonumber \\
& \leq & b(n) \sqrt{\sum_{l,r} \| T^L_l  B_{\phi}(r) \|_{\rho}^2} \cdot \sqrt{ \sum_{l,r} \| B_{\psi}(l) T^R_r \|_{\rho}^2 }. \label{psi_phi_CS}
\end{eqnarray}

The quantity under the first square root can be bounded in terms of $\Vert\phi\Vert$ as follows:
\beann
\sum_{l.r} \| T^L_l  B_{\phi}(r) \|_{\rho}^2 & = & \sum_{l,r} \Tr \rho B_{\phi}(r)^* (T_l^L)^* T_l^L B_{\phi}(r)  \\
&=& \sum_r  \Tr  \rho B_{\phi}(r)^* Q_L B_{\phi}(r) = \sum_r \| B_{\phi}(r) \|_{G_L -C_n}^2 \\
&\leq&  \frac{1}{1-b_{L}(n)}  \sum_r \Vert \Gamma_{G_L-C_n}(B_{\phi}(r))\Vert^2\\
&=& \frac{1}{1-b_{L}(n)} \Vert \phi \Vert^2\,,
\eeann
where we have used the definition of $Q_L$, Corollary~\ref{Cor:bd2} and \eq{phi1}.
The quantity under the second square root is similarly estimated in terms of  $\Vert \psi\Vert$:
\be
\sum_{l,r} \| B_{\psi}(l) T^R_r \|_{\rho}^2  \leq \frac{1}{1-b_{R}(n)} \Vert \psi \Vert^2\,.
\ee
Inserting these into \eqref{psi_phi_CS} yields
\be
\left| \langle \varphi, \psi \rangle -  \langle C_{\phi}, D_{\psi} \rangle_{\cK_G} \right|  \leq
\frac{b(n)}{\sqrt{(1-b_{L}(n))(1-b_{R}(n))}} \Vert \phi\Vert \Vert \psi \Vert\,.
\ee
\end{proof}

Now, we are ready to estimate the quantity of interest in \eq{innerproductformula}, which is an inner product of the form considered in Lemma~\ref{lem:innerproducts} with the additional
information that $\phi$ and $\psi$ are both orthogonal to $\cG_G$.

\begin{proposition}\label{prop:epsilonn_estimate}
Under the assumptions stated in Section \ref{sec:assumptions} and with the notations introduced there, we have the following estimate for the quantity $\epsilon_n$
defined in \eq{innerproductformula}:
\be\label{epsilon_bound}
\epsilon_n  \leq \frac{b(n)}{\sqrt{1-b_{LR}(n)}} +  \left(\frac{b(n)}{\sqrt{1-b_{LR}(n)}}\right)^2\left( 1+  b_G(n)
\right),
\ee
with
\be
b(n) = a(n)\Tr \rho^{-1},\ a(n) = \left\| \E^n - \ketbra{\idty}{\rho} \right\|,
\ee
and
$b_{LR}(n)$ and $b_G(n)$ are defined in \eqref{left_right_bn} and \eqref{b_LR}.
\end{proposition}

\begin{proof}
Any $\xi\in \cG_G$ belongs to both $\cG_{G_L-C_n} \otimes \cH_{G_R}$ and $\cH_{G_L} \otimes \cG_{C_n-G_R}$. Therefore,
there are unique matrices $B_{\xi}(r)$ and $B_{\xi}(l)$ and corresponding expressions  \eq{phi} and \eq{psi} for $\xi$. Since $\xi\in \cG_G$, there also exists $X\in \cK_G$ such that $\xi=\Gamma_G(X)$. By injectivity it follows that
\be
B^L_\xi(l) = T^L_l X,\quad B^R_\xi(r) =  X T^R_r.
\ee
Inserting the first relation above into the expression for $D_\psi$ and the second into $C_\phi$ we find the following special form of these matrices for a ground state $\xi$:
\bea
C_\xi &=& \sum_{r} B_{\phi}(r) \rho (T^R_{r})^*=  \sum_{r} X T^R_r \rho (T^R_{r})^*
= XQ_R \label{Cphi_perpG}\\
D_\xi &=& \sum_{l} (T^L_{l})^* B_{\psi}(l)=  \sum_{l} (T^L_{l})^* T^L_l X = Q_LX.\label{Dpsi_perpG}
\eea
We use this to extract information from the orthogonality of $\phi$ and $\psi$ to $\cG_G$. Using $\langle \phi,\xi\rangle = \langle \xi,\psi\rangle=0$, from Lemma \ref{lem:innerproducts} we have that
for all $X\in \Cx^{D_R\times D_L}$:
\bea
\left| \langle C_\phi, Q_LX \rangle_{\cK_G} \right| &\leq& \frac{b(n)}{\sqrt{1-b_{LR}(n)}} \| \phi \| \| \xi \|\label{CphiDpsi1}\\
\left| \langle XQ_R, D_\psi \rangle_{\cK_G} \right| &\leq& \frac{b(n)}{\sqrt{1-b_{LR}(n)}} \| \xi\| \| \psi \|.\label{CphiDpsi2}
\eea
Applying Corollary~\ref{Cor:bd2} gives
\[
\Vert \xi\Vert^2 = \Vert \Gamma_G(X)\Vert^2 \leq
 \left(1+  b_G(n) \right)
\Vert X\Vert_{G}^2.
\]
Using this with \eqref{CphiDpsi1} and \eqref{CphiDpsi2} yields
\bea\label{CphiDpsi-b}
\left| \langle C_\phi, Q_LX \rangle_{\cK_G} \right| &\leq& \delta(n) \Vert \phi \Vert \Vert X\Vert_{G}\\
\left| \langle XQ_R, D_\psi \rangle_{\cK_G} \right| &\leq& \delta(n) \Vert \psi \Vert \Vert X\Vert_{G},
\eea
where
$$
\delta(n) = \frac{b(n)}{\sqrt{1-b_{LR}(n)}}  \sqrt{1+  b_G(n)}.
$$
The LHS of these inequalities can be expressed in terms of the inner product  $\langle \cdot,\cdot\rangle_{G}$  as follows:
\beann
\langle C_\phi, Q_L X \rangle_{\cK_G} &=& 
\Tr C_\phi^* Q_L X =  \Tr Q_R Q_R^{-1} C_\phi^* Q_L X = \langle C_\phi Q_R^{-1}, X\rangle_{G}\\
\langle XQ_R, D_\psi \rangle_{\cK_G} &=& \Tr Q_R X^*D_\psi = \Tr Q_R X^* Q_L Q_L^{-1} D_\psi = \langle X, Q^{-1}_L D_\psi\rangle_{G}.
\eeann
The estimates \eq{CphiDpsi-b} now become: for all $X$
\bea\label{CphiDpsi-c}
\left| \langle C_\phi Q_R^{-1}, X\rangle_{G}\right| &\leq& \delta(n) \Vert \phi \Vert \Vert X\Vert_{G}\\
\left|  \langle X, Q^{-1}_L D_\psi\rangle_{G} \right| &\leq& \delta(n) \Vert \psi \Vert \Vert X\Vert_{G},
\eea
which imply
\bea\label{CphiDpsi-d}
\Vert C_\phi Q_R^{-1}\Vert_{G} &\leq& \delta(n) \Vert \phi \Vert \\
\Vert Q^{-1}_L D_\psi\Vert_{G} &\leq& \delta(n) \Vert \psi \Vert .
\eea
Noting the identity
$$
\langle C_\phi, D_\psi\rangle_{\cK_G} = \Tr Q_R (Q_R^{-1} C^*_\phi) Q_L (Q_L^{-1} D_\psi) = \langle  C_\phi Q_R^{-1}, Q^{-1}_L D_\psi\rangle_{G},
$$
we have
$$
\vert\langle C_\phi, D_\psi\rangle_{\cK_G} \vert \leq \Vert C_\phi Q_R^{-1}\Vert_{G}\Vert Q^{-1}_L D_\psi\Vert_{G}\leq  \delta(n)^2 \Vert \phi\Vert \Vert\psi\Vert.
$$
Moreover, it follows from \eqref{psi_phi_IP} that
\be
|\langle \phi,\psi\rangle| \leq  \vert\langle C_\phi, D_\psi\rangle_{\cK_G} \vert + \frac{b(n)}{\sqrt{1-b_{LR}(n)}} \Vert \phi\Vert \Vert\psi\Vert.
\ee
Combining the last two inequalities we obtain the final estimate
\be\label{innerproduct_bound}
|\langle \phi,\psi\rangle| \leq \left[ \frac{b(n)}{\sqrt{1-b_{LR}(n)}} +  \left(\frac{b(n)}{\sqrt{1-b_{LR}(n)}}\right)^2\left( 1+  b_G(n)\right) \right] \Vert \phi\Vert \Vert\psi\Vert.
\ee
\end{proof}

In the next section we verify the assumptions stated in this section for the AKLT model on the decorated honeycomb lattice with $n\geq 2$ and apply Proposition
\ref{prop:epsilonn_estimate} to show that for this model we have $\epsilon_n < 1/3$, for all $n\geq 3$.

\section{Gap of the decorated AKLT model}\label{sec:AKLT_Est}

In this section we prove Proposition~\ref{prop:epsilonn} and Theorem~\ref{thm:DecSpecGap} by applying the results of Section~\ref{sec:bound} to the decorated AKLT model discussed in Section~\ref{sec:models}. In this case, the graph $G$ is given by $Y_{v} \cup Y_{w}$ for two adjacent sites $v$ and $w$ in $\Gamma$. We decompose $G$ as $G = G_L - C_n - G_R$, where $C_n = Y_v \cap Y_w$, $G_L = Y_v \setminus C_n$ and $G_R = Y_w \setminus C_n$. The VBS (Valence Bond Solid) or PEPS (Product of Entangled Pairs) 
ground states on $G$ are depicted in Figure \ref{fig:VBSAKLT}. The corresponding Hilbert spaces are give by
\[
\cH_{G_L} = \cH_{G_R}= (\Cx^3\otimes \Cx^3)^{\otimes n}\otimes \Cx^4, \quad \text{and} \quad \cH_{C_n} = (\Cx^3)^{\otimes n}.
\]
For $s \in \{1, \, 3/2\}$, we use $\cB_s=\{\ket{s}, \, \ket{s-1} \, \ldots, \, \ket{-s}\}$ to denote an orthonormal basis of $\Cx^{2s+1}$ consisting of eigenvectors of the third component of spin associated to the spin-$s$ irreducible representation of $SU(2)$. We begin by discussing the MPS $\Gamma_{C_n}$ and its associated transfer operator $\E$. We then define the operator $\E_L$ associated with $G_L$ and prove that $\Gamma_{G_L-C_n}$ is injective for $n\geq 2$, after which we prove the main results.

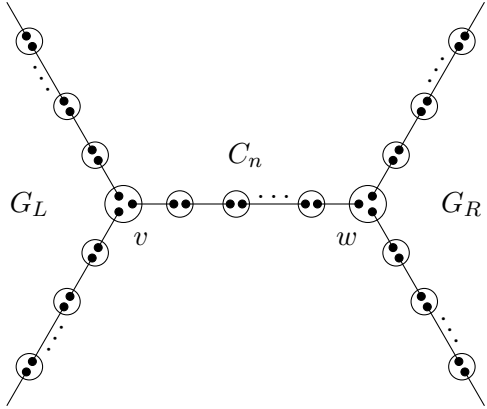
\begin{figure}[t]
	\begin{tikzpicture}
	\node[yshift = .65cm] at (0:1.625){$C_n$};
	\node at (0:-1.25){$G_L$};
	\node at (0:4.5){$G_R$};
	\draw (0:0) circle (7pt) node [below right, yshift=-.25cm] {$v$};
	\draw (0:3.25) circle (7pt) node [below left, yshift=-.25cm] {$w$};
	\foreach \n in {0,120,240}{
		\draw[fill=black](\n:.12) circle (1.5pt);
		\node[rotate = \n, yshift=3pt] at (\n:2) {$\ldots$};
		\draw(\n:.1)--(\n:.7);
		\draw(\n:.8)--(\n:1.4);
		\draw(\n:1.6)--(\n:2.4);
		\draw(\n:2.6)--(\n:3.1);
		\foreach \m in {.75,1.5,2.5} {
		\draw (\n:\m) circle (5pt);
		\draw[fill=black] (\n:{\m-.08}) circle (1.5pt);
	\draw[fill=black] (\n:{\m+.08}) circle (1.5pt);}
	}
\draw[fill=black](0:3.13) circle (1.5pt);
	\foreach \n in {60,-60}
	{\draw[fill=black,xshift=3.25cm](\n:.12) circle (1.5pt);
		\draw[xshift=3.25cm](\n:.1)--(\n:.7);
		\draw[xshift=3.25cm](\n:.8)--(\n:1.4);
		\draw[xshift=3.25cm](\n:1.6)--(\n:2.4);
		\draw[xshift=3.25cm](\n:2.6)--(\n:3.1);	
		\node[xshift=3.25cm, rotate = \n, yshift=3pt] at (\n:2) {$\ldots$};
		\foreach \m in {.75, 1.5, 2.5}{
	\draw[xshift=3.25cm] (\n:\m) circle (5pt);
	\draw[fill=black,xshift=3.25cm] (\n:{\m-.08}) circle (1.5pt);
	\draw[fill=black,xshift=3.25cm] (\n:{\m+.08}) circle (1.5pt);
}
	}
	\end{tikzpicture}
	\caption{The VBS picture for the decorated AKLT model.}
	\label{fig:VBSAKLT}
\end{figure}

On $C_n$ we have the one-dimensional AKLT spin chain, which has the bond dimension $D=2$. For this model, every physical spin-1 vertex is identified with the symmetric subspace of the virtual space $\Cx^2\otimes \Cx^2$. We will denote by $S^X, S^Y, S^Z$ the usual spin-1/2 operators, and by $S^{\pm}$ the corresponding lowering and raising operators. To differentiate between the physical and virtual spins, we will use $\ket{\ua},\ket{\da}\in\Cx^2$ to denote the standard orthonormal eigenbasis of $S^Z$ rather than $\ket{\pm 1/2}$. The intertwiner $P^{(1)}_{\text{sym}}:\Cx^2\otimes \Cx^2 \to \Cx^3$ that maps between the virtual and physical space of a site $v$ is given by:
\[
P^{(1)}_{\text{sym}} = \ketbra{1}{\ua\ua} + \ketbra{0}{\psi^+} + \ketbra{-1}{\da\da},\]
where $\ket{\psi^+}=\frac{1}{\sqrt{2}}(\ket{\ua\da}+\ket{\da\ua})$.

Recall that the symmetric subspace of $\Cx^2\otimes \Cx^2$ can be encoded into the MPS matrices
\[
P^{(1)}_{1} = \ketbra{\ua}{\ua}, \quad P^{(1)}_{0} = \sqrt{2} S^X, \quad P^{(1)}_{-1} = \ketbra{\da}{\da}.
\]
The ground states of the one-dimensional AKLT model can then be described as a valence-bond solid state obtained from projecting each (virtual) edge of the graph into the singlet states $\ket{\psi^-}=\frac{1}{\sqrt{2}}(\ket{\ua\da}-\ket{\da\ua})$; this is represented by the MPS matrix
\[ K = \frac{1}{\sqrt{2}}\left( \ketbra{\ua}{\da} - \ketbra{\da}{\ua} \right) = \sqrt{2}i S^Y.\]
With these matrices, and with a convenient choice of normalization, the ground state space of the one-dimensional AKLT matrix is given by
\[
\Gamma_{C_n}(B) = \sum_{i_1, \ldots, \, i_n \in \{\pm1,0\}} \Tr[BV_{i_n}\ldots V_{i_1}] \ket{i_1\ldots i_n}
\]
where $V_i = \frac{2}{\sqrt{3}} K P^{(1)}_i$. Explicitly,
\begin{equation}\label{1d-aklt-mps}
V_{1} = -\sqrt{\frac{2}{3}} S^{+},
\quad V_0 = \frac{2}{\sqrt{3}}S^{Z},
\quad V_{-1} = \sqrt{\frac{2}{3}} S^{-}.
\end{equation}
Given the form of $\Gamma_{C_n}$, the choice of multiplying on the \emph{left} by $K$ in the definition of $V_{i}$ corresponds to projecting the edge to the \emph{right} of the associated site into the singlet state. This convention will also be used to define the tensors $T_\ell^L$ and $T_r^R$. For more details on this and other MPS constructions, see \cite{S11, PVWC07, O14}.

Noting that $S^+ B S^- + S^+ B S^- = 2(S^XBS^X + S^Y B S^Y)$, the transfer operator $\E:M_2\to M_2$ associated with $\Gamma_{C_n}$ takes the form
\begin{equation}
\E(B) = \sum_{i\in \{\pm 1, \, 0\}}V_i^*BV_i = \frac{4}{3}(S^X B S^X + S^Y B S^Y + S^Z B S^Z),
\end{equation}
which can be easily diagonalized as
\begin{equation}\label{hE}
\E = \ketbra{\idty}{\rho} - \frac{2}{3}\sum_{U\in\{X,Y,Z\}} \ketbra{S^U}{S^U},
\end{equation}
where $\rho = \idty/2$ is the maximally mixed state. This allows to easily compute
\[
\E^n = \ketbra{\idty}{\rho} + 2\frac{(-1)^n}{3^n} \sum_{U\in\{X,Y,Z\}} \ketbra{S^U}{S^U},
\]
from which \eqref{an} takes the explicit form
\[
a(n) = \left\| \E^n - \ketbra{\idty}{\rho} \right\| =
3^{-n},
\]
and $b(n) = \Tr(\rho^{-1}) a(n) = 4\cdot 3^{-n}$. By Corollary \ref{Cor:bd1}, this implies that $\Gamma_{C_n}$ is injective when $n \ge 2$. It can easily be shown that it is not injective for $n=1$.

We now consider $G_L$ and $G_R$. For the decorated AKLT model, $D_L = D_R = 4$ and so $\cK_G=M_4$. We first construct the operator $\E_L$ associated with $G_L$, and use this to prove $\Gamma_{G_L-C_n}$ is injective of $n\geq2$. The analogous operator $\E_R$ for $G_R$ and the injectivity of of $\Gamma_{C_n-G_R}$ follow from similar calculations.

We first note that $G_L$  can be written as $[u^1_1,u^1_n]\times [u^2_1,u^2_n]\times\{v\}$, where the sites $u^i_k$ correspond to the $2n$ spin-1's, and $v$ is the spin-$3/2$. By grouping the sites $u_i^1$ and $u_i^2$ into a single site $(u_i^1,u_i^2)$, we can recognize the ground states of $H_{G_L}$ as a PEPS. We choose the product basis for $\cH_{G_L}$ given by
\[
\ket{i_1,j_1}\otimes\dots\otimes\ket{i_n,j_n}\otimes \ket{k} \quad i_1,\dots, i_n, j_1,\dots, j_n \in \{ \pm 1,0 \}, \; k \in \left\{ \pm \frac 32, \pm \frac 12\right \}.
\]

For each element $\ket{l}_L=\ket{i_1,j_1}\otimes \dots \otimes \ket{i_n,j_n}\otimes \ket{k}$ of the basis, the $2\times 4$ matrix $T^L_l$ is given by
\[
T^L_l = W_k^L V_{i_n} \otimes V_{j_n} \cdots V_{i_1} \otimes V_{j_1},
\]
where the $V_i$ are as defined in \eqref{1d-aklt-mps}, and the $W_k^L \in \cL(\Cx^4,\Cx^2)$ are given by the PEPS representation of the AKLT on the hexagonal lattice, which we now define. Analogous to the spin-1 case, the virtual space of a spin-3/2 particle is the symmetric subspace of three spin-1/2 particles, and so the intertwiner $P^{(3/2)}_{\text{sym}}:\Cx^2\otimes \Cx^2 \otimes \Cx^2 \to \Cx^4$  between the virtual and physical space is given by
\[
P^{(3/2)}_{\text{sym}} = \ketbra{3/2}{\ua\ua\ua} + \ketbra{1/2}{\phi^+} + \ketbra{-1/2}{\phi^-} + \ketbra{-3/2}{\da\da\da},
\]
where
\bea
\ket{\phi^+} &= \frac{1}{\sqrt{3}}(\ket{\ua\ua\da}+\ket{\ua\da\ua}+\ket{\da\ua\ua})
= \frac{1}{\sqrt{3}}\ket{\da}\ket{\ua\ua} + \sqrt{\frac{2}{3}}\ket{\ua}\ket{\psi^+}
, \\
\ket{\phi^-} &= \frac{1}{\sqrt{3}}(\ket{\ua\da\da}+\ket{\da\ua\da}+\ket{\da\da\ua})=
\frac{1}{\sqrt{3}}\ket{\ua}\ket{\da\da} + \sqrt{\frac{2}{3}}\ket{\da}\ket{\psi^+}
.
\eea
By grouping two virtual edges to the left of $v$, see Figure~\ref{fig:VBSAKLT}, the virtual space can be represented by the MPS matrices $P_{k}^{(3/2)}\in M_{2\times 4}$ defined by
\beann
P^{(3/2)}_{3/2} = \ketbra{\ua}{\ua\ua}, \quad & P^{(3/2)}_{1/2} = \frac{1}{\sqrt{3}}\ketbra{\da}{\ua\ua} + \sqrt{\frac{2}{3}}\ketbra{\ua}{\psi^+} \\
P^{(3/2)}_{-3/2} = \ketbra{\da}{\da\da}, \quad & P^{(3/2)}_{-1/2} = \frac{1}{\sqrt{3}}\ketbra{\ua}{\da\da} + \sqrt{\frac{2}{3}}\ketbra{\da}{\psi^+} .
\eeann
Once again projecting edges on the \emph{right} of $v$ into a singlet state (and choosing a convenient normalization) we define $W_k^L = \sqrt{2}K P_{k}^{(3/2)}$. Explicitly,
\bea
W_{3/2}^L = -\ketbra{\da}{\ua\ua},
&
W_{1/2}^L = \phantom{-} \frac{1}{\sqrt{3}}\ketbra{\ua}{\ua\ua} - \sqrt{\frac{2}{3}} \ketbra{\da}{\psi^+},
\\
W_{-3/2}^L = \ketbra{\ua}{\da\da},
&
W_{-1/2}^L = -\frac{1}{\sqrt{3}} \ketbra{\da}{\da\da} + \sqrt{\frac{2}{3}}\ketbra{\ua}{\psi^+},
\eea
which satisfies $\sum_{i=1}^4 W_i^L(W_i^L)^* = 2\idty_{\Cx^2}$. For $B\in M_2$, define $\E^{\Yright}(B) = \sum_{i} (W_i^L)^* B W_i^L$. While this is a completely positive map from $M_2$ to $M_4$, unlike the MPS case it is not unital, since
\[\E^{\Yright}(\idty) = \sum_{i=1}^4 (W_i^L)^*W_i^L = \frac{4}{3}(\ketbra{\ua\ua}{\ua\ua} + \ketbra{\da\da}{\da\da} + \ketbra{\psi^+}{\psi^+}) =  \idty + \frac{4}{3} \mathbf{S}\cdot \mathbf{S},\]
where as usual $\mathbf{S}=(S^X,S^Y,S^Z)$ and $\mathbf{S}\cdot\mathbf{S} = S^X\otimes S^X + S^Y\otimes S^Y + S^Z\otimes S^Z$. By direct calculation, we see that
\be \label{EYRight}
(\E^{\Yright})^*(B) =
c(B) \rho + \sum_{U\in \{X,Y,Z\}} c_U(B) S^U,
\ee
where
\beann
c(B) & = & \frac{4}{3}(\expval{\ua\ua}{B}{\ua\ua} + \expval{\da\da}{B}{\da\da} + \expval{\psi^+}{B}{\psi^+}) \\
c_X(B) & = &  -\frac{2\sqrt{2}}{3} {\rm Re}[\expval{\psi^+}{B}{\ua\ua} +  \expval{\da\da}{B}{\psi^+} ]\\
c_Y(B) & = & -\frac{2\sqrt{2}}{3} {\rm Im}[\expval{\psi^+}{B}{\ua\ua} +  \expval{\da\da}{B}{\psi^+}]\\
c_Z(B) & = & \frac{2}{3} (\expval{\da\da}{B}{\da\da}-\expval{\ua\ua}{B}{\ua\ua}).
\eeann
It can easily be checked that $(\E^{\Yright})^t\circ \tau = (\E^{\Yright})^t$ where $\tau:M_4\to M_4$ is the transposition operator
\[
\tau(A\otimes B) = B\otimes A
\]
Combining this with \eqref{EYRight} allows us to verify that
\begin{equation*}
(\E^{\Yright})^t(\rho \otimes \rho) = \rho, \quad
(\E^{\Yright})^t(S^U \otimes S^{U'}) = \delta_{U,U'} \frac{1}{3} \rho, \quad
(\E^{\Yright})^t(\rho \otimes S^U) = (\E^{\Yright})^t(S^U \otimes \rho) = -\frac{1}{3} S^U,
\end{equation*}
or equivalently
\begin{equation}
\E^{\Yright} = \ketbra{\idty\otimes \idty}{\rho} + \frac{4}{3}
\ket{\mathbf{S}\cdot \mathbf{S}}\bra{\rho}
- \frac{4}{3} \sum_{U\in\{X,Y,Z\}} \left(\ket{S^U \otimes \idty}+\ket{\idty\otimes S^U} \right)\bra{S^U}.
\end{equation}

To simplify notation we define
\[
\ket{\Omega^U} = \ket{S^U \otimes \idty}+\ket{\idty\otimes S^U}\  \forall U \in \{X,Y,Z\},
\]
and notice that
\[
\norm{\mathbf{S}\cdot \mathbf{S}}_2=\frac{\sqrt{3}}{2}, \quad \norm{\Omega^U}_2=\sqrt{2}, \quad \text{and} \quad \braket{\mathbf{S}\cdot \mathbf{S}}{\Omega^U}_2 = 0
\]
for all $U$. The transfer matrix for $G_L$, defined in \eqref{transfer-matrix-L-R}, is then given by $\E_L = (\E^n \otimes \E^n) \circ \E^{\Yright}$, which can be simplified to
\bea
\E_L & = & \left(\ketbra{\idty}{\rho} + 2\frac{(-1)^n}{3^n} \sum_{U} \ketbra{S^U}{S^U} \right) \otimes \left(\ketbra{\idty}{\rho} + 2\frac{(-1)^{n}}{3^n} \sum_{U} \ketbra{S^U}{S^U} \right) \E^{\Yright} \nonumber\\
& = &
\ketbra{\idty\otimes\idty}{\rho}
+ 2\frac{(-1)^{n+1}}{3^{n+1}} \sum_{U} \ketbra{\Omega^U}{S^U}
+ \frac{4}{3^{2n+1}} \ketbra{\mathbf{S}\cdot \mathbf{S}}{\rho}.
\eea
Using this decomposition to compute $Q_L$ gives
\begin{equation}
Q_L = \idty + \frac{4}{3^{2n+1}} \mathbf{S}\cdot \mathbf{S}, \implies \spec(Q_L) = \left\{1-\frac{1}{3^{2n}}, \, 1+\frac{1}{3^{2n+1}}\right\}.
\end{equation}
Therefore, $q_L = 1-\frac{1}{3^{2n}}$, and moreover, since $\E_L$ is a completely positive map,
\[
\|\E_L\| = \|Q_L\| = 1+\frac{1}{3^{2n+1}}.
\]
Since $Q_L$ is invertible, the theory of Section~\ref{sec:bound} applies and we can use the above relations to prove the following result.
\begin{lemma}\label{AKLT_GL_inj}
$\Gamma_{G_L-C_n}$ is injective for $n\geq 2$.
\end{lemma}

\begin{proof} Let $B\in M_D$ and consider $\Gamma_{G_L-C_n}(B).$ Applying Corollary~\ref{Cor:bd2} gives
\[
    \|\Gamma_{G_L-C_n}(B)\|^2 \geq \left(1-b_L(n)\right) \|B\|_{G_L-C_n}^2
\]
with $b_L(n) = \frac{4 a(n)\|\E_L\|}{\rho_{\rm min}q_L}$.
By inserting the values of $a(n)$, $\|\E_L\|$, $\rho_{\min}$, and $q_L$  into this expression, one finds that
	\[
		1-b_L(n) = 1-\frac{8(1+3^{-2n-1})}{3^n(1-3^{-2n})}.
	\]
	This quantity is strictly positive for any $n\geq 2$ from which it follows that $\Gamma_{G_L-C_n}$ is injective.
\end{proof}

We now consider $G_R$. The operator $\E_R$ is obtained using a similar construction as $\E_L$. As with $G_L$, we can once again can group the spin-1 particles into pairs and to construct an orthonormal basis
\[
\ket{r} = \ket{i_1,j_1}\otimes\dots\otimes\ket{i_n,j_n}\otimes \ket{k} \quad i_1,\dots, i_n, j_1,\dots, j_n \in \{ \pm 1,0 \}, \; k \in \left\{ \pm \frac 32, \pm \frac 12\right \},
\]
for which the corresponding tensor is given by
\be
T_r^R = V_{i_n} \otimes V_{j_n} \cdots V_{i_1} \otimes V_{j_1}  W_k^R,
\ee
where $V_{1}, \, V_0$ and $V_{-1}$ are as before, and $W_{k}^R = 2 K\otimes K(P_k^{(3/2)})^*$. Explicitly,
\bea
W_{3/2}^R = \ketbra{\da\da}{\ua},
&
W_{1/2}^R = \frac{1}{\sqrt{3}} \ketbra{\da\da}{\da} - \sqrt{\frac{2}{3}}\ketbra{\psi^+}{\ua},
\\
W_{-3/2}^R = \ketbra{\ua\ua}{\da},
&
W_{-1/2}^R =  \frac{1}{\sqrt{3}}\ketbra{\ua\ua}{\ua} - \sqrt{\frac{2}{3}} \ketbra{\psi^+}{\da}.
\eea
Similar to the case of $\E^L$, we have $\E^{R} := \E^{\Yleft}\circ (\E^n\otimes\E^n)$ where $\E^{\Yleft}:M_4\to M_2$ is defined by
\be \label{LRTranspose}
\E^{\Yleft}(B) := \sum_{i} (W_i^R)^* B W_i^R = \sum_{i} W_i^L B (W_i^L)^* = (\E^{\Yright})^t(B).
\ee
The final equality above follows from recognizing
\bea
W_{3/2}^R & = \phantom{-}(W_{-3/2}^L)^*, \quad W_{1/2}^R = -(W_{-1/2}^L)^*\\
W_{-3/2}^R& = -(W_{3/2}^L)^*, \quad W_{-1/2}^R = \phantom{-}(W_{1/2}^L)^*.
\eea
It follows from the analogous arguments as used in Lemma~\ref{AKLT_GL_inj} above that $\Gamma_{C_n-G_R}$ is also injective for $n\geq 2$. We can now prove Proposition~\ref{prop:epsilonn}, and Theorem~\ref{thm:DecSpecGap}.

\begin{proof}[Proof of Proposition~\ref{prop:epsilonn} and Theorem~\ref{thm:DecSpecGap}]
Since $\E=\E^t$, from \eqref{LRTranspose} it follows that
\[
\E_R = [(\E^n\otimes\E^n)\circ\E^{\Yright}]^t = (\E_L)^t.
\]
Therefore, $\|\E_R\| =\|\E_L\|$ and
\be
Q_R := (\E_R)^t(\rho) = \E_L(\rho) = \frac{1}{2}Q_L.
\ee
As a consequence,
\[
q_R = \frac{1}{2}q_L = \rho_{\min} q_L.
\]
Using \eqref{epsilon_bound} to estimate \eqref{innerproductformula}, we find
\be
\label{AKLT_ep_bound}
\epsilon_n
\leq  \frac{4a(n)}{\sqrt{1-b_{LR}(n)}} + \left(\frac{4a(n)}{\sqrt{1-b_{LR}(n)}}\right)^2\left(1+ b_G(n)\right).
\ee
From \eqref{left_right_bn} and the values above, it is clear that $b_L(n) = b_R(n)$, and so
\bea
1-b_{LR}(n) = (1-b_L(n))(1-b_R(n)) &= \left( 1- \frac{8(1+3^{-2n-1})}{3^n(1-3^{-2n})}\right)^2, \nonumber\\
1+b_G(n) &= \left(1+ \frac{8(1+3^{-2n-1})^2}{3^n(1-3^{-2n})^2}\right).
\eea
This establishes Proposition~\ref{prop:epsilonn}.

Inserting these into \eqref{AKLT_ep_bound}, we find that $\epsilon_n<\frac{1}{3}$ whenever $n\geq 3$. By \eqref{gap_cond} this implies that the decorated AKLT model has a positive spectral gap above the ground state energy for $n\geq 3$. This completes the proof of Theorem~\ref{thm:DecSpecGap}.
\end{proof}
\section{Discussion}

We proved an explicit positive lower bound for the spectral gap above the ground state of the AKLT model on the decorated honeycomb lattice for $n\geq 3$, where $n$ is the number of
vertices inserted on each edge of the honeycomb lattice. It is natural to ask whether the approach of this paper could be used to prove that the AKLT model on the honeycomb lattice itself
($n=0$) is gapped too, which is expected. It is clear to us, however, that significant changes to the arguments would be necessary to achieve this. For example, a numerical calculation
shows that $\epsilon_1 \sim .478 > 1/3$. Therefore, our method does not work for $n=1$. For the case $n=2$, we do not have a good estimate of $\epsilon_2$, but it is conceivable that our approach could be extended to the case $n=2$. For the model with $n=3$, however, we proved a positive lower bound. By using a numerically calculated value for the gap for the small system on $Y$, which appears in \eq{comparable}, ($\gamma_Y\sim 0.2966$),  and the rigorous estimate showing $\epsilon_3 < 0.2683$ (Proposition \ref{prop:epsilonn}), 
we found the following uniform lower bound for the gap: $\gamma  > 0.0289$.

About generalizations to frustration-free models on other decorated lattices on the other hand, we can be rather optimistic. For example, we expect that similar arguments will work to
study the spectral of AKLT models on decorated hypercubic lattices of any dimension. One could also try to apply our approach to some of the more exotic hybrid valence bond models discussed in
\cite{WHR14}.

For physical reasons, one wants the spectral gap to be robust under small perturbations of the interactions. It seems very likely that the AKLT models on the decorated honeycomb lattices (and likely also on the honeycomb lattice itself) satisfy the Local Topological Quantum Order condition introduced by Bravyi, Hastings, and Michalakis \cite{BHM10}. If so, the stability theorem of Michalakis and Zwolak \cite{MZ13} would apply to the AKLT models on decorated lattices and provide the desired robustness of the spectral gap.

\section*{Acknowledgments}

This work arose from discussions during the follow-up workshop on \emph{Gapped Ground State Phases of Quantum Many-Body Systems} to the \emph{2018 Arizona School of Analysis and Mathematical Physics}, organized by Robert Sims and two of the authors (H.~A. and A.~Y.) 
and supported by NSF Grant DMS-1800724.
A.~L. acknowledges support from the Walter  Burke  Institute  for Theoretical Physics in the form of the Sherman Fairchild Fellowship as well as support from the Institute for Quantum  Information  and  Matter  (IQIM),  an  NSF  Physics Frontiers Center (NFS Grant PHY-1733907).
B.~N. acknowledges support by the National Science Foundation under Grant DMS-1813149 and a CRM-Simons
Professorship for a stay at the Centre de Recherches Math\'ematiques (Montr\'eal) during Fall 2018, where part of this work was carried out.

\end{document}